\documentclass[lettersize,journal]{IEEEtran}
\usepackage{amsmath,amsfonts,amssymb}
\usepackage{amsthm}
\usepackage{algorithmic}
\usepackage{algorithm}
\usepackage{array}
\usepackage{multirow}
\usepackage[caption=false,font=normalsize,labelfont=sf,textfont=sf]{subfig}
\usepackage{textcomp}
\usepackage{stfloats}
\usepackage{url}
\usepackage{verbatim}
\usepackage{graphicx}
\usepackage[dvipsnames]{xcolor}
\usepackage{cite}
\newtheorem{proposition}{Proposition}

\begin{document}
\bstctlcite{BSTcontrol}

\title{Exposing Noise-Deficiency Signatures of Deep Learning-Based False Data Injection Attacks Through AC State Estimation}

\author{Yichen~Liu,~\IEEEmembership{Student~Member,~IEEE,}
Hongyu~Wu,~\IEEEmembership{Senior~Member,~IEEE,}
and~Bo~Liu,~\IEEEmembership{Member,~IEEE}
\thanks{This work has been submitted to the IEEE for possible publication. Copyright may be transferred without notice, after which this version may no longer be accessible.}
}

\maketitle

\begin{abstract}

Deep learning (DL)-based false data injection (FDI) attacks can resemble normal measurements and evade both conventional bad data detection (BDD) and measurement-space data-driven detectors. This paper identifies a previously overlooked signature of such attacks in state-estimation (SE) residual space. The residuals of reconstruction-based DL attacks exhibit deficient stochastic variation, manifested as abnormally small residual energy and reduced directional dispersion. We establish the conditions under which this signature emerges and explain why the attacks carrying it can evade conventional BDD. We show that SE can reveal residual-space differences obscured by operating-state variation in measurement space. We then explain the origin of the signature through the noise-suppression mechanism of DL attack generation. A generator can learn the reproducible structure of historical measurements but cannot reproduce stochastic measurement components that are unpredictable from its input. This signature enables label-free detection using a simple two-sided test on the conventional BDD statistic, while the full residual vector can also be exploited when labeled attack samples are available. Experiments on the IEEE 14-bus and 118-bus AC systems show that the label-free test detects up to 99\% of three representative DL attacks at a false-alarm rate comparable to that of BDD, which detects at most 5\%. Ablation studies further examine the effects of residual representation, decision rule, and measurement-noise magnitude and distribution.

\end{abstract}


\section{Introduction}

\IEEEPARstart{M}{odern} AC power systems rely on supervisory control and data acquisition (SCADA) systems and networked sensing technologies for monitoring and control \cite{zhang_smart_2021}. This connectivity expands the attack surface of state estimation (SE), whose estimates support contingency analysis, optimal power flow, and economic dispatch. False data injection (FDI) attacks can corrupt these estimates and mislead downstream decisions. This paper investigates an estimator-side defense that uses measurements and residuals already available in the SE pipeline, without requiring additional infrastructure or grid modifications such as secured communication channels, phasor measurement units, or active topology reconfiguration \cite{bi_graphical_2014,pei_pmu_2020,liu_optimal_2021}.

The standard estimator-side defense is residual-based bad data detection (BDD), developed within the classical weighted-least-squares (WLS) SE framework. BDD evaluates discrepancies between observed and model-predicted measurements, typically using a $\chi^2$ test on the weighted residual energy and a largest-normalized-residual (LNR) test on individual residual components \cite{monticelli_reliable_1983}. Detection thresholds are calibrated under the assumed noise model to control the false-alarm probability, and BDD raises an alarm when a residual statistic exceeds its threshold. Sensor failures, calibration errors, and isolated bad telemetry can inflate residuals and trigger detection. However, corruption that does not produce sufficiently large residual statistics can evade conventional BDD.

Liu et al. \cite{liu_false_2011} showed that an attacker with knowledge of the system Jacobian can construct FDI attacks that bypass BDD, commonly referred to as \emph{classical FDI attacks}. Subsequent work applied this construction to operational targets including price manipulation \cite{xie_false_2010}, automatic generation control \cite{tan_modeling_2017}, and multi-period dispatch \cite{zhang_can_2018}. Despite their stealthiness against BDD, these attacks can shift measurements away from the normal operating distribution. Measurement-space data-driven (DD) detectors exploit such distributional changes using supervised classifiers or unsupervised anomaly detection on raw measurements or their principal components \cite{esmalifalak_detecting_2017}, as well as deep models that capture spatial \cite{zhang_spatial-temporal_2022}, temporal \cite{zhao_lstm-based_2020}, or measurement reconstruction patterns \cite{sun_power_2022}. Thus, BDD targets residual-inflating corruption, whereas measurement-space DD detectors can identify residual-stealthy attacks that distort the measurement distribution.

Classical FDI attacks, however, face two practical limitations. First, their construction generally requires system information that may be difficult for an attacker to obtain, such as topology, line parameters, or sensor configuration, although partial-knowledge and sparse-access variants relax these requirements \cite{song_intelligent_2018,hao_sparse_2015}. Second, the measurement-distribution shifts they induce can be exploited by DD detectors. Advances in deep learning, particularly generative adversarial networks (GANs) \cite{goodfellow_gan_2014}, have opened another route, in which attacks are learned from historical SCADA measurements instead of being constructed from an explicit system model. We refer to the resulting attacks as \emph{deep learning (DL) attacks}. These attacks learn generation or reconstruction rules and seek to keep compromised measurements close to the normal-operation distribution monitored by measurement-space DD detectors \cite{liu_matrix-completion-based_2023,costilla-enriquez_attack_2023,liu_gan-based_2022,ahmadian_cyber_2018,zhao_controllable_2025,li_false_2024}. DL attacks thus aim to achieve stealthiness while reducing dependence on explicit system knowledge.

Early DL attacks pursued this goal primarily with GANs under DC models on 5- to 30-bus systems \cite{ahmadian_cyber_2018,mohammadpourfard_cgan_2020,shahriar_iattackgen_2021,liu_lsgan_2024}. However, directly generated attacks did not necessarily evade BDD. In \cite{shahriar_iattackgen_2021}, only about 55\% of GAN-generated attacks passed BDD, and an additional refinement procedure was needed to achieve near-complete stealthiness. Research subsequently expanded DL attacks from small systems with linear DC models to larger systems with nonlinear AC models.

This expansion introduces two additional challenges, the nonlinear AC measurement mapping and the increasing dimensionality of the measurement space. Recent methods have addressed these challenges on systems as large as the IEEE 118-bus system \cite{jiao_new_2021,bhattacharjee_deep_2022,costilla-enriquez_attack_2023}, often by conditioning attack generation on real-time measurements \cite{jiao_new_2021,costilla-enriquez_attack_2023} or state information \cite{bhattacharjee_deep_2022,rahman_adversarial_2024}. Such conditioning, however, makes attack construction dependent on real-time information and potentially sensitive to distribution shifts between training and deployment conditions. More recently, representation-learning GANs (RL-GANs) \cite{liu_representation_2026} combine an autoencoder with latent-space generation to learn compact representations of historical AC measurements, enabling attack generation on both IEEE 14-bus and 118-bus systems without real-time measurements as generator inputs.

Despite substantial advances in DL attack construction, defenses specifically targeting these attacks have received comparatively limited attention. Much of the existing literature has focused on demonstrating that DL attacks can evade conventional BDD and measurement-space DD detectors \cite{liu_matrix-completion-based_2023,costilla-enriquez_attack_2023,liu_representation_2026}. Several studies present experimental results showing unusually low BDD statistics for DL-generated attacks, yet the statistical implications of these results remain largely unexplored \cite{jiao_new_2021,rahman_adversarial_2024,liu_representation_2026,xiao_guaranteed_2025}. Evading existing detectors does not necessarily imply statistical indistinguishability from normal measurements in the SE residual space. It therefore remains unclear whether DL attacks exhibit systematic and exploitable statistical signatures, why these signatures arise from the attack-generation process, and whether they can support detection without labeled attack samples or prior knowledge of the attack generator.

This paper addresses these questions through the following contributions.

\begin{itemize}
    \item[1)] We identify a previously overlooked residual-space signature of DL attacks that evade conventional BDD and measurement-space DD detectors. Under the conditions established here, their SE residuals exhibit deficient stochastic variation, manifested as low residual energy and reduced directional dispersion. SE reveals this deficiency, which operating-state variation can obscure in measurement space. The signature enables statistical tests calibrated using only normal data, without labeled attacks or knowledge of the attack generator.

    \item[2)] We explain the statistical origin of this signature through conditional-mean reconstruction. An ideal squared-loss generator preserves predictable measurement components while suppressing stochastic variation that is unpredictable from its input. We show how this suppression propagates through SE and reduces residual variation and, when residual bias is small, residual energy. This explains why DL attacks can evade upper-tail BDD while remaining statistically distinguishable from normal measurements in residual space.

    \item[3)] We evaluate the signature on the IEEE 14-bus and 118-bus AC systems against three DL attacks that evade conventional BDD and a measurement-space DD detector trained on classical FDI attacks. The results validate the signature and demonstrate its detectability using a label-free two-sided residual-energy test and supervised classifiers on full residual vectors. Ablation studies examine the effects of residual representation, decision rule, and measurement-noise magnitude and distribution.

\end{itemize}

The remainder of this paper is organized as follows. Section~\ref{sec:preliminaries} reviews AC SE and BDD and defines the DL attacks and the threat model. Section~\ref{sec:defense} characterizes the residual-space signature, explains its origin, and shows how to exploit it. Section~\ref{sec:experiments} reports the experimental evaluation. Section~\ref{sec:conclusion} concludes the paper.

\section{System and Attack Model}
\label{sec:preliminaries}

\subsection{AC State Estimation}
\label{subsec:SE}
In modern AC power systems, SCADA-telemetered measurements are streamed to the control center, where a state estimator based on the AC power-flow equations recovers the operating state. Let $z \in \mathbb{R}^{m}$ denote the measurement vector, which collects the standard SCADA channels of bus active and reactive power injections $P_i, Q_i$, branch active and reactive power flows $P_{ij}, Q_{ij}$, and bus voltage magnitudes $|V_i|$, where $i$ indexes buses and $(i,j)$ indexes branches. Let $x \in \mathbb{R}^{n}$ denote the state vector. The two are related through the measurement model
\begin{equation}
    z = h(x) + e,
    \label{eq:meas-model}
\end{equation}
where $h(\cdot):\mathbb{R}^{n} \rightarrow \mathbb{R}^{m}$ is the nonlinear AC power-flow mapping, with coefficients fixed by the topology and line parameters, and $e \in \mathbb{R}^{m}$ is additive measurement noise modeled as zero-mean Gaussian with covariance $\Sigma = \mathrm{diag}(\sigma_1^2, \dots, \sigma_m^2)$. SE is solved as a WLS problem,
\begin{equation}
    \hat{x} = \arg\min_{x} \, [z - h(x)]^{\top}\, \Sigma^{-1}\, [z - h(x)],
    \label{eq:wls}
\end{equation}
which the nonlinearity of $h(\cdot)$ makes iterative, typically via Gauss--Newton. The estimated state $\hat{x}$ then feeds downstream energy management system (EMS) applications such as contingency analysis, optimal power flow, and economic dispatch.

\subsection{Bad Data Detection}
\label{subsec:BDD}
The post-estimation residual,
\begin{equation}
    r(\hat{x}) = z - h(\hat{x}),
    \label{eq:residual}
\end{equation}
captures the discrepancy between the observed measurement and its model-reconstructed value, and is the quantity on which BDD is built. BDD applies statistical tests to this residual to flag inconsistencies that the WLS fit cannot absorb into noise. The most widely used test is the $\chi^2$ test, which evaluates the weighted residual norm
\begin{equation}
    J(\hat{x}) = r(\hat{x})^{\top} \Sigma^{-1} r(\hat{x}).
    \label{eq:chi-stat}
\end{equation}
Under the Gaussian noise assumption, $J(\hat{x})$ approximately follows a $\chi^2$ distribution with $m - n$ degrees of freedom, so a snapshot is flagged whenever
\begin{equation}
    J(\hat{x}) > \tau, \quad \tau = \chi^{2}_{m-n,\,1-\alpha},
    \label{eq:bdd-rule}
\end{equation}
where $\alpha$ is the prescribed false-alarm probability.

\subsection{DL Attacks}
\label{subsec:attack-benchmarks}

A DL attack replaces the legitimate measurement vector with one synthesized by a learned model,
\begin{equation}
    z^{a} = G(u;\phi),
    \label{eq:dl-attack}
\end{equation}
where $G(\cdot;\phi)$ is trained on historical SCADA measurements of the target system and $u$ is its input. This differs from additive FDI, $z^{a} = z + a$, which perturbs the legitimate measurement.

\textit{Threat model.} The attacker is positioned between the substation remote terminal units and the control-center EMS and can replace the SCADA telemetry with synthesized values. It has historical SCADA records of the target system but no access to the system model, line parameters, or the operator's SE solver. Its goal is to bias the state estimate while evading BDD and any deployed measurement-space DD detector. We consider attacks generated as their original methods prescribe, and attackers who adapt to a residual-space detector are beyond the scope of this paper.

\textit{Representative DL attacks.} We study three such attacks, a measurement-completion attack based on a deep neural network (DNN) \cite{rahman_adversarial_2024}, a reconstruction attack based on an autoencoder (AE) \cite{costilla-enriquez_attack_2023}, and the RL-GAN attack \cite{liu_representation_2026}. The DNN maps the first $n_{\mathrm{bus}}$ measurement channels of $z_{t}$ to the full measurement vector and is trained with a squared reconstruction loss. The AE reconstructs $z_{t}$ through an encoder--decoder bottleneck with the same loss. The RL-GAN encodes $z_{t}$ with an autoencoder, trains a GAN in the latent space with an $L_{1}$ reconstruction term on the latent code, and decodes the nearest generated code back to measurement space. Although their architectures and training procedures differ, all three generate $z^{a}$ through a mapping trained with a squared reconstruction loss, and Section~\ref{subsec:origin} examines how this shared property produces the residual-space signature.

\section{Residual-Space Signature and Defense}
\label{sec:defense}

This section identifies the signature that DL attacks leave in SE residuals (Section~\ref{subsec:signature}), characterizes the conditions under which it emerges (Section~\ref{subsec:conditions}), explains its origin in DL-based attack generation (Section~\ref{subsec:origin}), and shows how it can be exploited for detection (Section~\ref{subsec:defense}).

\subsection{Residual-Space Signature}
\label{subsec:signature}

DL attacks are trained to synthesize measurements that resemble historical data. Several studies report unusually low BDD statistics for such attacks \cite{jiao_new_2021,rahman_adversarial_2024,liu_representation_2026,xiao_guaranteed_2025}. The attacks studied here exhibit a similar pattern. They rarely trigger BDD, while an existing measurement-space DD detector trained on classical FDI attacks detects them inconsistently (Section~\ref{subsec:evasion-matrix}). Their SE residuals, however, reveal a systematic departure from normal behavior. The BDD statistic $J$ falls well below the normal range (Fig.~\ref{fig:residual-norm}), and the residual covariance exhibits a substantially lower effective rank (Table~\ref{tab:signature}). These abnormally small and concentrated residuals reveal a deficiency in residual-space variation, which we term the \emph{residual-space signature} of DL attacks.

For the attacks studied here, the signature has two manifestations, abnormally low residual energy and reduced directional dispersion. The former is captured by the scalar BDD statistic $J$, while the latter is reflected in the covariance structure of the full residual vector $r$. Both provide complementary information for detection (Section~\ref{subsec:defense}).

\subsection{Conditions for the Signature}
\label{subsec:conditions}

\textit{Local residual model.} Let $T(z)=z-h(\hat{x}(z))$ denote the residual map, where $\hat{x}(z)$ is the WLS estimate in \eqref{eq:wls}. Around an operating point $x_0$, linearize $h(x_0+\delta)\approx h(x_0)+H\delta$, where $H=\partial h/\partial x|_{x_0}$ is the Jacobian. This local linearization serves only to characterize residual covariances analytically, and the experiments in Section~\ref{sec:experiments} use nonlinear WLS. Assume $\Sigma\succ0$ and that $H$ has full column rank. The WLS problem then has the local solution $\hat{\delta}=K(z-h(x_0))$, with gain $K=(H^\top\Sigma^{-1}H)^{-1}H^\top\Sigma^{-1}$, yielding
\begin{equation}
    r \approx S\big(z-h(x_0)\big), \qquad S=I-HK.
    \label{eq:linear-residual-map}
\end{equation}
The matrix $S$ is the $\Sigma^{-1}$-orthogonal projector with kernel $\mathrm{Col}(H)$ and therefore satisfies
\begin{equation}
    SH=0.
    \label{eq:SH0}
\end{equation}
Thus, every first-order measurement variation attributable to a change in system state is removed from the residual.

\textit{Normal and attacked residuals.} Under normal operation, $z=h(x_0+\delta)+e$ with $\mathbb{E}[e]=0$ and $\mathrm{Cov}(e)=\Sigma$. Equations~\eqref{eq:linear-residual-map} and \eqref{eq:SH0} give
\begin{equation}
    r_N\approx Se, \qquad \mathrm{Cov}(r_N)=S\Sigma S^\top.
    \label{eq:normal-residual}
\end{equation}
The state variation $\delta$, which dominates the variation in raw measurements, therefore does not enter $r_N$. Similarly, an attacked measurement $z^a=G(c)$ can be decomposed locally as $z^a\approx h(x_0)+H\delta_a+\tilde{e}_a$, where $\tilde{e}_a$, with mean $\mu_a$ and covariance $\Sigma_a$, is the remaining variation in the synthesized measurement, including retained meter-noise variation and reconstruction errors. It is not meter noise. This decomposition is not unique, but since $SH=0$, only $S\tilde{e}_a$ enters the residual, and it does not depend on how $\delta_a$ is chosen. Hence,
\begin{equation}
    r_a\approx S\tilde{e}_a, \qquad \mathrm{Cov}(r_a)=S\Sigma_aS^\top.
    \label{eq:attack-residual}
\end{equation}

\textit{Noise-deficiency condition.} Since $SH=0$, the signature is determined by the remainder $\tilde{e}_a$ of the generator output. An attack is \emph{noise-deficient along} a direction $v\in\mathbb{R}^m$ with $S^\top v\neq0$ if it satisfies the noise-deficiency condition
\begin{equation}
    v^\top S\Sigma_aS^\top v
    <
    v^\top S\Sigma S^\top v,
    \label{eq:condition}
\end{equation}
i.e., if the remainder of the attack exhibits less variation than the meter noise along a residual-sensitive direction. Under the local model, the two sides of \eqref{eq:condition} equal $\mathrm{Var}(v^\top r_a)$ and $\mathrm{Var}(v^\top r_N)$, so the condition is necessary and sufficient for reduced residual variation along $v$. The stronger covariance ordering $S\Sigma_aS^\top\preceq S\Sigma S^\top$, with strict inequality along at least one direction, establishes an overall reduction in residual covariance. Proposition~\ref{prop:bdd-order} further shows that, under its Gaussianity and zero-bias assumptions, this ordering is sufficient to keep the upper-tail BDD detection probability at or below that of normal measurements.

\begin{proposition}
\label{prop:exposure}
Under the local model \eqref{eq:linear-residual-map}, suppose that, for some choice of $\delta_a$, the attack is noise-deficient along $v$ and reproduces the normal operating-state variation, such that $\mathrm{Cov}(\delta_a)=\mathrm{Cov}(\delta)=\Sigma_x$. Assume further that the state and stochastic terms are uncorrelated for both normal and attacked measurements. Then $\mathrm{Var}(v^\top r_a)<\mathrm{Var}(v^\top r_N)$. In measurement space, for any direction $w$,
\begin{equation}
    \frac{\mathrm{Var}(w^\top z^a)}
         {\mathrm{Var}(w^\top z)}
    =
    \frac{s_w+w^\top\Sigma_aw}
         {s_w+w^\top\Sigma w},
    \qquad
    s_w=w^\top H\Sigma_xH^\top w,
    \label{eq:z-masking}
\end{equation}
which approaches $1$ as $s_w$ dominates both $w^\top\Sigma w$ and $w^\top\Sigma_aw$.
\end{proposition}

\textit{Proof:} The residual-space inequality follows from \eqref{eq:normal-residual}, \eqref{eq:attack-residual}, and \eqref{eq:condition}. For the measurement-space ratio, linearize the measurements as $z=h(x_0)+H\delta+e$ and $z^a=h(x_0)+H\delta_a+\tilde{e}_a$. The stated covariance relation and uncorrelatedness assumptions yield \eqref{eq:z-masking}. \hfill$\blacksquare$

Proposition~\ref{prop:exposure} explains why the residual-space signature can be obscured in measurement space. Under the stated assumptions, the attack reproduces the operating-state variation $H\Sigma_xH^\top$, which can dominate the difference in noise covariance. Consequently, the variance ratio in \eqref{eq:z-masking} approaches $1$ along directions dominated by state variation, although this does not imply identical measurement distributions. In residual space, \eqref{eq:SH0} removes the first-order state variation and exposes the noise deficiency. This role of SE is analogous to the innovation in Kalman-style estimators \cite{netto_robust_2018} and structured residuals in model-based fault diagnosis \cite{gertler_structured_2000}. Whether DL attacks satisfy \eqref{eq:condition} is examined next.

\subsection{Origin of the Signature}
\label{subsec:origin}

Section~\ref{subsec:conditions} establishes the conditions under which the residual-space signature arises. We now explain why DL attacks can satisfy condition~\eqref{eq:condition}. The key distinction is that normal measurements contain both state-dependent variation and stochastic meter noise, whereas a reconstruction-based generator may fail to reproduce meter noise that is unpredictable from its input. We characterize this effect using the population-optimal solution to squared-loss reconstruction.

Each DL attack in Section~\ref{subsec:attack-benchmarks} generates measurements from an effective input $c$ through a mapping trained with squared reconstruction loss. Here $c$ is the input $u$ for the DNN and the low-dimensional latent code for the AE and RL-GAN. For the RL-GAN, the mapping is the autoencoder decoder, with adversarial training selecting the code $c$. Under squared loss, the population-optimal generator over all measurable mappings, evaluated on the training input distribution, is the conditional mean:
\begin{equation}
    G^\star(c)
    =
    \mathbb{E}[z\mid c]
    =
    \mathbb{E}[h(x)\mid c]+\mathbb{E}[e\mid c].
    \label{eq:cond-mean}
\end{equation}
Under the local approximation $h(x)\approx h(x_0)+H\delta$ with the same Jacobian $H$, $\mathbb{E}[h(x)\mid c]\approx h(x_0)+H\,\mathbb{E}[\delta\mid c]$, and \eqref{eq:SH0} eliminates this state-dependent component, leaving
\begin{equation}
    r_a^\star
    \approx
    S\big(H\,\mathbb{E}[\delta\mid c]+\mathbb{E}[e\mid c]\big)
    =
    S\,\mathbb{E}[e\mid c].
    \label{eq:attack-residual-cm}
\end{equation}
The law of total covariance gives
\begin{equation}
    \Sigma
    =
    \mathrm{Cov}\big(\mathbb{E}[e\mid c]\big)
    +
    \mathbb{E}\big[\mathrm{Cov}(e\mid c)\big].
    \label{eq:total-variance}
\end{equation}
Consequently, within the local model,
\begin{equation}
    \mathrm{Cov}(r_a^\star)
    =
    S\Sigma S^\top
    -
    S\,\mathbb{E}\big[\mathrm{Cov}(e\mid c)\big]\,S^\top
    \preceq
    \mathrm{Cov}(r_N).
    \label{eq:deficit}
\end{equation}
The deficit term $\mathbb{E}[\mathrm{Cov}(e\mid c)]$ is the meter-noise covariance that cannot be predicted from $c$, which conditional-mean reconstruction therefore suppresses. Thus, condition~\eqref{eq:condition} follows from the reconstruction mechanism along every direction $v$ satisfying $v^\top S\,\mathbb{E}[\mathrm{Cov}(e\mid c)]\,S^\top v>0$.

The three attacks exhibit this mechanism to different extents. For the DNN, $c$ contains only the first $n_{\mathrm{bus}}$ measurement channels. With the Gaussian noise in \eqref{eq:meas-model} independent across channels and of the state, the noise in the remaining channels cannot be predicted from $c$, so condition~\eqref{eq:condition} holds along every residual direction to which this noise contributes. For the AE, $c$ is a low-dimensional code derived from the full noisy measurement. Compression to this code may suppress meter-noise variation, but compression alone does not guarantee this effect, which is evaluated empirically (Section~\ref{subsec:residual-mechanism}). The RL-GAN may exhibit similar behavior through its autoencoder decoder, but its adversarially selected codes can depart from the training distribution. Thus, the conditional-mean analysis explains the mechanism behind the observed signature and does not guarantee it for every attack.

A trained generator generally differs from the population optimum. With $G=G^\star+\Delta$, the attacked residual becomes $r_a\approx r_a^\star+S\Delta(c)$. The projected reconstruction error $S\Delta$ can shift the residual mean and add residual covariance, which may offset the deficiency induced by noise suppression. The deficiency is therefore a tendency under the stated assumptions, not a guarantee for every trained generator. Fig.~\ref{fig:epoch} illustrates this behavior. Early in training, reconstruction errors place $J$ above the normal range, and as reconstruction improves, $J$ decreases and eventually falls below it.

The resulting deficiency also explains why upper-tail BDD can miss these attacks under additional assumptions.

\begin{proposition}
\label{prop:bdd-order}
Under the local residual model \eqref{eq:linear-residual-map}, suppose that $S\mu_a=0$, $S\Sigma_aS^\top\preceq S\Sigma S^\top$, and both normal and attacked residuals are Gaussian. Then $\Pr(J_a>\tau)\le\Pr(J_N>\tau)$ for every $\tau\ge0$. If the threshold $\tau$ is calibrated such that $\Pr(J_N>\tau)=\alpha$, the attack detection probability under BDD is at most $\alpha$.
\end{proposition}

\textit{Proof:} Since $S\Sigma_aS^\top\preceq S\Sigma S^\top$, the normal residual has the same distribution as $r_a+y$, where $y$ is an independent, possibly degenerate, zero-mean Gaussian vector with covariance $S\Sigma S^\top-S\Sigma_aS^\top$. The event $J\leq\tau$ defines a symmetric convex set in residual space. Anderson's inequality \cite{anderson_integral_1955} therefore gives $\Pr(J_N\leq\tau)\leq\Pr(J_a\leq\tau)$ for every $\tau$. \hfill$\blacksquare$

Proposition~\ref{prop:bdd-order} establishes that, under its assumptions, the detection probability of upper-tail BDD cannot exceed the normal false-alarm rate. An upper-tail test on $J$ cannot flag abnormally small residual energy. However, the ideal residual $r_a^\star$, although zero-mean when $\mathbb{E}[e]=0$, need not be Gaussian, so this result does not automatically apply to reconstruction-based attacks.

Two qualifications are worth noting. First, residual bias may offset the covariance reduction and increase the expected BDD statistic. Second, for an unbiased attack, uniform covariance reduction lowers the expected residual energy without changing the effective rank. The reduced effective ranks in Table~\ref{tab:signature} therefore reflect nonuniform changes in residual covariance.

\subsection{Signature-Based Defense}
\label{subsec:defense}

The residual-space signature provides a basis for detecting DL attacks that evade upper-tail BDD. It can be exploited through two representations, the scalar statistic $J$ in \eqref{eq:chi-stat}, which summarizes residual energy, and the full residual vector $r\in\mathbb{R}^m$, which retains directional information. Table~\ref{tab:defense-grid} classifies the detectors by representation and by training data. A label-free detector uses normal data only, whereas a supervised detector also requires labeled attacks. As Fig.~\ref{fig:residual-norm} illustrates, the reduction in $J$ is pronounced enough for a simple statistical test, so we use simple detectors instead of developing more complex models.

\begin{table}[!t]
\caption{Residual-space detectors by representation and training data.}
\label{tab:defense-grid}
\centering
\renewcommand{\arraystretch}{1.15}
\footnotesize
\begin{tabular*}{\linewidth}{@{\extracolsep{\fill}}l l l}
\hline\hline
 & Label-free (normal data) & Supervised (labeled attacks) \\
\hline
Vector $r$ & Open (future work) & Classifier on $r$ \\
Statistic $J$ & Two-sided test \eqref{eq:two-sided} & Threshold learned on $J$ \\
\hline\hline
\end{tabular*}
\end{table}

Previous studies have also used SE residuals for FDI detection, including distributional-similarity tests on weighted least absolute value residuals \cite{cheng_highly_2022}, neural networks using residual elements or features reduced by principal component analysis \cite{tabakhpour_neural_2019}, and adaptive Mahalanobis scoring of regression residuals \cite{nagaraj_state_2021}. Unlike the present approach, these methods do not explicitly target the stochastic variation deficiency of DL attacks.

\textit{Residual-energy detection.} On $J$, we use a label-free two-sided test. Let $J_1,\ldots,J_{N_v}$ be the residual-energy statistics of $N_v$ normal validation snapshots. A snapshot is flagged if
\begin{equation}
    J<\tau_{\mathrm{lo}}
    \quad\text{or}\quad
    J>\tau_{\mathrm{hi}},
    \label{eq:two-sided}
\end{equation}
where $\tau_{\mathrm{lo}}$ and $\tau_{\mathrm{hi}}$ are the empirical $\alpha/2$ and $1-\alpha/2$ quantiles of $\{J_i\}$, respectively. The test uses the same statistic and the same target false-alarm rate $\alpha$ as BDD but differs in its rejection region and its calibration. BDD places the whole rate in the upper tail of the nominal $\chi^2$ distribution, whereas the two-sided test splits it equally between the two tails of the empirical normal distribution. The upper threshold therefore retains sensitivity to excessive residual energy, while the lower threshold captures the identified signature. Its advantage is expected for attacks concentrated in the lower tail, not for every attack distribution. Empirical calibration avoids reliance on the nominal $\chi^2$ distribution, which may be inaccurate under nonlinear SE or nonideal measurement noise, and the realized false-alarm rates are verified on normal test data in Section~\ref{sec:experiments}.

\textit{Residual-vector detection.} The full vector $r$ retains directional and cross-channel information that $J$ discards. When labeled attacks are available, a supervised classifier on $r$ can learn this information, and Section~\ref{sec:experiments} evaluates such a classifier. A label-free detector on $r$ would need to respond to concentrated residuals as well as to conventional outliers. Standard outlier scores, which flag samples far from the normal data, may instead rank concentrated attack residuals as more normal than normal observations, and Section~\ref{subsec:ablation} examines this behavior for three such detectors. A label-free vector detector that is sensitive to concentration is left for future work. Both evaluated detectors use only the measurement model, noise covariance, and SE solver already required by BDD. The two-sided test needs only normal data, whereas the classifier also requires labeled attacks.

\section{Experiments}
\label{sec:experiments}

This section validates the identified residual-space signature and evaluates its detectability on the IEEE 14-bus and 118-bus AC systems against three DL attacks: DNN \cite{rahman_adversarial_2024}, AE \cite{costilla-enriquez_attack_2023}, and RL-GAN \cite{liu_representation_2026}. Section~\ref{subsec:exp-setup} presents the experimental setup and evaluation protocol. Section~\ref{subsec:evasion-matrix} evaluates the attacks against existing defenses. Section~\ref{subsec:residual-mechanism} demonstrates the label-free detectability of the signature and verifies and quantifies it. Section~\ref{subsec:ablation} examines the effects of residual representation, decision rule, and measurement-noise conditions.

\subsection{Experimental Setup}
\label{subsec:exp-setup}

\textit{Systems and measurement models.} All experiments use the IEEE 14-bus and 118-bus AC test systems. The measurement vector $z$ comprises randomly selected SCADA channels, including bus active and reactive power injections $\{P_i,Q_i\}$, branch active and reactive power flows $\{P_{ij},Q_{ij}\}$, and bus voltage magnitudes $\{|V_i|\}$. The number of selected channels is determined by a fixed measurement-redundancy coefficient, with the same measurement configuration maintained across snapshots, attacks, and random seeds. The resulting measurement and state dimensions are $(m,n)=(68,27)$ for the 14-bus system and $(591,235)$ for the 118-bus system.

\textit{Normal-operation data.} Normal operating conditions are generated using measured load profiles from the ERCOT 15-minute interval archive. The profiles are scaled to each system's nominal aggregate demand and proportionally distributed across load buses. For each system, $8000$ operating snapshots are generated by sampling the load profiles, solving the AC power flow with MATPOWER \cite{zimmerman_matpower_2020}, and forming $z_t=h(x_t)+e_t$, where $x_t=(|V_t|,\theta_t)$ and $e_t\sim\mathcal{N}(0,\Sigma)$, with each $\sigma_i$ set to $1\%$ of the nominal magnitude of channel $i$. These snapshots support attack-model training, residual calibration, and detector evaluation. Model hyperparameters are reported in Table~\ref{tab:attack-config}.

\textit{Data partitioning.} For each system, the normal snapshots are divided into $5600$ training, $1200$ validation, and $1200$ test samples, and each attack set provides a training pool and $1200$ test samples. The test set is reserved for final evaluation, and the validation set is used only for threshold calibration and hyperparameter tuning. Detector training labels are sampled only from the training partition.

\begin{table}[!t]
\caption{Attack model configurations.}
\label{tab:attack-config}
\centering
\renewcommand{\arraystretch}{1.15}
\footnotesize
\setlength{\tabcolsep}{4pt}
\begin{tabular}{l l c c}
\hline\hline
Attack & Component & 14-bus & 118-bus \\
\hline
\multirow{3}{*}{DNN}
 & Input channels & $14$ & $118$ \\
 & Hidden layers & \multicolumn{2}{c}{$300, 600, 600, 300, 100$} \\
 & Learning rate / batch & \multicolumn{2}{c}{$10^{-6}$ / $1000$} \\
\hline
\multirow{4}{*}{AE}
 & Encoder hidden & \multicolumn{2}{c}{$150, 300, 100$} \\
 & Decoder hidden & \multicolumn{2}{c}{$100, 300, 150$} \\
 & Latent dimension & $1$ & $1$ \\
 & Learning rate / batch & \multicolumn{2}{c}{$10^{-5}$ / $1000$} \\
\hline
\multirow{4}{*}{RL-GAN}
 & AE latent dimension & $1$ & $2$ \\
 & Generator / discriminator & \multicolumn{2}{c}{$200, 500, 300, 100$} \\
 & Learning rate / batch & \multicolumn{2}{c}{$10^{-4}$ / $1000$} \\
 & Loss & \multicolumn{2}{c}{Binary cross-entropy $+\,0.5\,L_{1}$} \\
\hline\hline
\end{tabular}
\end{table}

\begin{table}[!t]
\caption{Classifier hyperparameters.}
\label{tab:classifier-config}
\centering
\renewcommand{\arraystretch}{1.2}
\footnotesize
\setlength{\tabcolsep}{4pt}
\begin{tabular}{l l l}
\hline\hline
Classifier & 14-bus & 118-bus \\
\hline
SVM (RBF) & $C{=}0.5$, $\gamma{=}$scale & $C{=}0.1$, $\gamma{=}$scale \\
XGBoost & $4$ trees, depth $2$, lr $0.01$ & $4$ trees, depth $2$, lr $0.01$ \\
MLP & $(128,256,128)$, lr $10^{-3}$ & $(1000,2000,1000)$, lr $10^{-3}$ \\
\hline\hline
\end{tabular}
\end{table}

\textit{Detectors.} Two existing defenses, BDD and a classical-FDI-trained DD detector, are used to assess whether the DL attacks evade current detection, and the two-sided test is used to exploit the residual-space signature. BDD applies the $\chi^{2}$ test \eqref{eq:bdd-rule} at $\alpha=0.05$, with thresholds $\chi^{2}_{41,0.95}=56.94$ on 14-bus and $\chi^{2}_{356,0.95}=401.0$ on 118-bus. The classical-FDI-trained DD detector \cite{liu_false_2011} is trained on normal snapshots and on classical FDI attacks, which bias one state by $5\%$ on 14-bus and two states by $5\%$ on 118-bus. It is then evaluated on each DL attack, none of which it has seen during training.

\textit{Ablation detectors.} The ablations in Section~\ref{subsec:ablation} add two kinds of detectors. Three standard outlier detectors, an Isolation Forest (IF), a one-class support vector machine (SVM), and an autoencoder, are applied to $r$. They are trained on the $5600$ normal training snapshots and raise an alarm above the 95th percentile of their normal validation scores. The supervised detectors are binary classifiers on $z$ or on $r$, trained separately for each system and attack with $100$ labels per class drawn from the training pool. Each uses a radial basis function (RBF)-kernel SVM, gradient-boosted decision trees (XGBoost), or a three-layer multilayer perceptron (MLP), with hyperparameters fixed in advance (Table~\ref{tab:classifier-config}). Each supervised result is the mean over five seeds that resample the labels and reinitialize the classifiers. The classical-FDI-trained DD detector uses the same three classifier types.

\textit{Evaluation metrics.} Detection is reported by the attack detection rate (Det), and the false-alarm rates on normal test snapshots are stated in the text. Supervised detectors are reported by test $F_{1}$ with attacks as the positive class.

\subsection{Attack Effectiveness Against Existing Defenses}
\label{subsec:evasion-matrix}

\begin{table}[!t]
\caption{Attack detection by BDD and the classical-FDI-trained DD detector.}
\label{tab:evasion-matrix}
\centering
\renewcommand{\arraystretch}{1.05}
\footnotesize
\begin{tabular*}{\linewidth}{@{\extracolsep{\fill}}l l c c c c}
\hline\hline
\multirow{2}{*}{Sys.} & \multirow{2}{*}{Attack} & BDD & \multicolumn{3}{c}{Classical-FDI-trained DD ($F_{1}$)} \\
\cline{4-6}
& & Det & MLP & SVM & XGB \\
\hline
\multirow{4}{*}{\rotatebox[origin=c]{90}{14-bus}}
 & \emph{Classical FDI (ref.)} & --  & \emph{1.00} & \emph{1.00} & \emph{0.99} \\
 & DNN    & \textbf{0.03} & \textbf{0.00} & 0.63 & 0.59 \\
 & AE     & \textbf{0.01} & \textbf{0.00} & 0.68 & 0.52 \\
 & RL-GAN & \textbf{0.05} & \textbf{0.00} & \textbf{0.00} & 0.56 \\
\hline
\multirow{4}{*}{\rotatebox[origin=c]{90}{118-bus}}
 & \emph{Classical FDI (ref.)} & --  & \emph{0.99} & \emph{1.00} & \emph{0.96} \\
 & DNN    & \textbf{0.04} & \textbf{0.00} & \textbf{0.00} & 0.69 \\
 & AE     & \textbf{0.05} & \textbf{0.00} & \textbf{0.00} & 0.60 \\
 & RL-GAN & \textbf{0.02} & \textbf{0.00} & \textbf{0.00} & 0.56 \\
\hline\hline
\end{tabular*}
\end{table}

\textit{BDD.} BDD is nearly blind to all three DL attacks. Its detection rate stays between $0.01$ and $0.05$ on both systems (Table~\ref{tab:evasion-matrix}), at or near its false-alarm rate on normal measurements ($0.045$ on 14-bus and $0.044$ on 118-bus).

\begin{figure}[!t]
\centering
\includegraphics[width=0.85\linewidth]{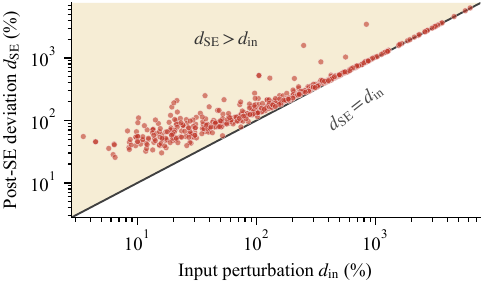}
\caption{Post-SE deviation versus input perturbation for an AE attack on 14-bus.}
\label{fig:impact}
\end{figure}

\textit{Attack impact.} Fig.~\ref{fig:impact} shows an example AE attack on 14-bus. Let $\hat{z}=h(\hat{x}(z))$ and $\hat{z}^{a}=h(\hat{x}(z^{a}))$ denote the SE-estimated measurements for the normal and attacked measurements. The input perturbation $d_{\mathrm{in}}$ and the post-SE deviation $d_{\mathrm{SE}}$ are the averages of $|z^{a}_{i}-z_{i}|/|\hat{z}_{i}|$ and $|\hat{z}^{a}_{i}-\hat{z}_{i}|/|\hat{z}_{i}|$ over all measurement channels $i$. Every snapshot has $d_{\mathrm{SE}}>d_{\mathrm{in}}$, with a median ratio of $1.38$, so the attack shifts the SE output at least as much as it perturbs the measurements.

\textit{Classical-FDI-trained DD.} This detector does not transfer from classical FDI to DL attacks. It reaches $F_{1} \geq 0.96$ on the classical FDI it was trained for, but its MLP scores $0.00$ on all six attacks, and its SVM scores $0.00$ on every 118-bus attack and on 14-bus RL-GAN (Table~\ref{tab:evasion-matrix}). XGBoost retains only partial detection, with $F_{1}$ between $0.52$ and $0.69$. The existing defenses therefore cover the DL attacks inconsistently. This raises the question of whether the attacks leave a signature that neither defense tests for, which Section~\ref{subsec:residual-mechanism} examines in residual space.

\subsection{Residual-Space Signature and Detectability}
\label{subsec:residual-mechanism}

\begin{table}[!t]
\caption{Detection rates of BDD and the two-sided test.}
\label{tab:twosided}
\centering
\renewcommand{\arraystretch}{1.10}
\footnotesize
\begin{tabular*}{\linewidth}{@{\extracolsep{\fill}}l l c c}
\hline\hline
& & BDD & Two-sided \\
Sys. & Attack & Det & Det \\
\hline
\multirow{3}{*}{14-bus} & DNN & 0.03 & 0.93 \\
 & AE & 0.01 & 0.89 \\
 & RL-GAN & 0.05 & 0.92 \\
\hline
\multirow{3}{*}{118-bus} & DNN & 0.04 & 0.99 \\
 & AE & 0.05 & 0.99 \\
 & RL-GAN & 0.02 & 0.99 \\
\hline\hline
\end{tabular*}
\end{table}

\subsubsection{Detectability}
We first evaluate whether the identified signature enables label-free detection using the two-sided test \eqref{eq:two-sided}. The test flags residual energy $J$ outside the 2.5th and 97.5th percentiles of normal $J$, calibrated from 1200 normal validation samples per system. Because it operates on the conventional BDD statistic without attack labels or a learned detector, its performance directly reflects the detectability of the signature. Table~\ref{tab:twosided} compares its detection rate with that of conventional BDD on the held-out test set. The test also adds negligible computation. On one thread of an Intel Core i9-13900KF CPU, it takes less than $1~\mu$s per snapshot, compared with $45.5$~ms for the 118-bus WLS that BDD already requires.

All six attack configurations are detectable without attack labels. The two-sided test achieves detection rates of $0.89$--$0.99$, with false-alarm rates of $0.040$ and $0.041$ on the 14-bus and 118-bus systems, respectively. Conventional BDD, which applies a one-sided nominal $\chi^2$ threshold to the same statistic, detects at most $0.05$ of the attacks. The lower threshold $\tau_{\mathrm{lo}}$ accounts for most detections, flagging $0.87$--$0.98$ of the attacks, whereas the upper threshold flags at most $0.05$, comparable to the BDD detection rate. These results demonstrate that the residual-energy statistic already contains a strong attack signature that conventional upper-tail BDD does not exploit.

\begin{figure*}[!t]
\centering
\includegraphics[width=0.85\linewidth]{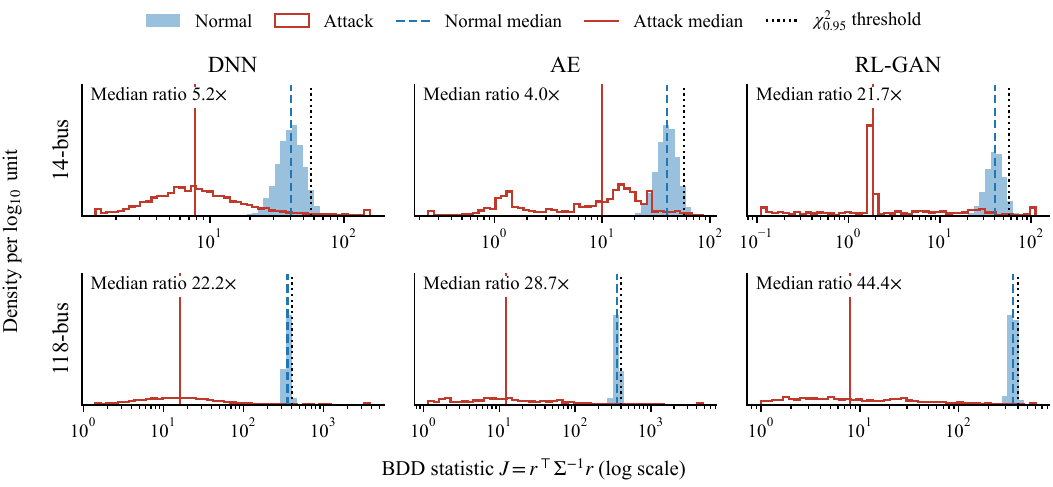}
\caption{Distribution of the BDD statistic $J$.}
\label{fig:residual-norm}
\end{figure*}

\subsubsection{Signature}
Fig.~\ref{fig:residual-norm} illustrates the statistical departure underlying this detection advantage. It compares the distributions of $J$ for normal and attacked measurements across all six configurations, with the conventional BDD threshold $\chi^2_{0.95}$ indicated.

The attacked measurements consistently produce smaller residual energy than normal measurements. In every configuration, the attack distribution is concentrated toward lower $J$, with its median reduced by a factor of $4.0$--$44.4$ (Fig.~\ref{fig:residual-norm}). This downward departure contrasts with the elevated residual energy targeted by conventional BDD and explains why attacks that pass its upper-tail test remain detectable through the lower tail.

\begin{figure*}[!t]
\centering
\includegraphics[width=0.85\linewidth]{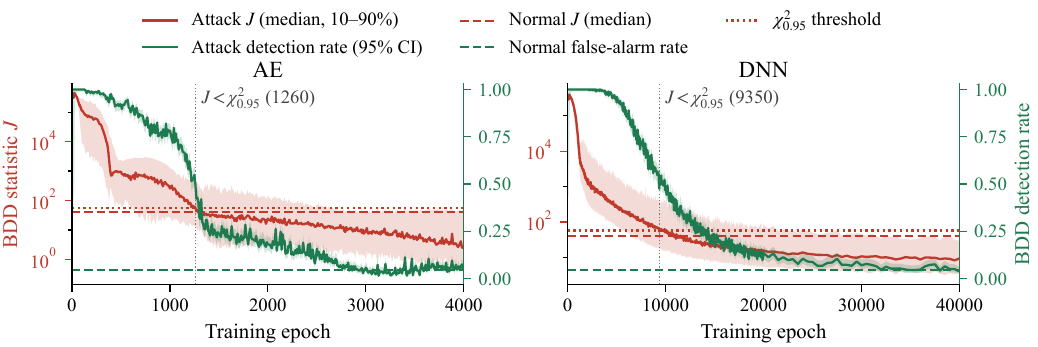}
\caption{BDD statistic $J$ and BDD detection rate of the attacks during training on 14-bus.}
\label{fig:epoch}
\end{figure*}

\subsubsection{Training Dynamics}
We examine how the residual-space signature develops during AE and DNN training on the 14-bus system and whether early stopping allows an attack to evade both BDD and the two-sided test. Fig.~\ref{fig:epoch} tracks residual energy and BDD detection across epochs. Table~\ref{tab:epoch} compares both tests at five representative stages spanning the threshold-crossing transition, using 400 attacked samples per stage.

As training progresses, attack residual energy shifts from abnormally high to abnormally low values, while BDD detection declines (Fig.~\ref{fig:epoch}). The median attack $J$ crosses the BDD threshold at epochs $1260$ (AE) and $9350$ (DNN), although BDD still detects $0.50$ and $0.48$ of the attacks, respectively, due to residual-energy dispersion. At epochs $2830$ (AE) and $35000$ (DNN), BDD detection falls to $0.04$ and $0.05$, reaching the normal false-alarm rate. The corresponding median $J$ values are $10.5$ and $9.7$, well below the normal 5th percentile of $27.0$.

Two-sided detection remains effective at both ends of training, with a temporary decline near the threshold crossing (Table~\ref{tab:epoch}). Initially, it detects $0.86$ (AE) and $0.93$ (DNN) of the attacks, predominantly through upper-tail anomalies. Detection falls to minima of $0.61$ (AE) and $0.66$ (DNN) after the crossing, then rises as lower-tail anomalies dominate. When BDD detection reaches the normal false-alarm rate, two-sided detection reaches $0.75$ (AE) and $0.89$ (DNN), increasing to $0.86$ and $0.90$ at the last epochs in Fig.~\ref{fig:epoch}. The observed trajectories reveal no early-stopping window for evading both tests, since two-sided detection remains at least $0.73$ whenever BDD detection is at or below $0.06$, consistent with the mechanism in Section~\ref{subsec:origin}.

These dynamics also explain the attack selection in Tables~\ref{tab:evasion-matrix} and \ref{tab:twosided}. Achieving BDD evasion requires continued training until detection approaches the false-alarm rate, by which point most attack $J$ values fall below the normal range. Early stopping within this range is impractical because it is narrow ($25.0\le J\le59.7$ on 14-bus) and attack $J$ values are widely dispersed. Even at minimum two-sided detection, only $0.39$ (AE) and $0.34$ (DNN) of attacks fall within this range, while BDD still detects $0.18$ and $0.37$, respectively. The 14-bus attacks in these tables come from separate training runs at the later epochs $4700$ (AE) and $88000$--$98900$ (DNN), explaining minor differences from Table~\ref{tab:epoch}. The 118-bus attacks use the first checkpoint at which the validation BDD pass rate reaches that of normal measurements.

\begin{table}[!t]
\caption{Detection rates of the 14-bus attacks during training.}
\label{tab:epoch}
\centering
\renewcommand{\arraystretch}{1.10}
\footnotesize
\setlength{\tabcolsep}{2pt}
\begin{tabular*}{\linewidth}{@{\extracolsep{\fill}}l l c c c c}
\hline\hline
Attack & Milestone & Epoch & Median $J$ & BDD & Two-sided \\
\hline
\multicolumn{2}{l}{Normal (ref.)} & -- & 40.2 & 0.04 & 0.04 \\
\hline
\multirow{5}{*}{AE} & Early training & 1000 & 256.0 & 0.78 & 0.86 \\
 & Median $J<\chi^2_{0.95}$ & 1260 & 56.9 & 0.50 & 0.69 \\
 & Median $J<$ normal median & 1300 & 34.9 & 0.33 & 0.69 \\
 & BDD at false-alarm rate & 2830 & 10.5 & 0.04 & 0.75 \\
 & Last epoch in Fig.~\ref{fig:epoch} & 4000 & 2.6 & 0.05 & 0.86 \\
\hline
\multirow{5}{*}{DNN} & Early training & 5000 & 286.7 & 0.93 & 0.93 \\
 & Median $J<\chi^2_{0.95}$ & 9350 & 53.3 & 0.48 & 0.69 \\
 & Median $J<$ normal median & 10850 & 38.2 & 0.36 & 0.68 \\
 & BDD at false-alarm rate & 35000 & 9.7 & 0.05 & 0.89 \\
 & Last epoch in Fig.~\ref{fig:epoch} & 40000 & 9.3 & 0.04 & 0.90 \\
\hline\hline
\end{tabular*}
\end{table}

\begin{table}[!t]
\caption{Residual energy and effective rank of normal and attacked residuals.}
\label{tab:signature}
\centering
\renewcommand{\arraystretch}{1.10}
\footnotesize
\setlength{\tabcolsep}{2pt}
\begin{tabular*}{\linewidth}{@{\extracolsep{\fill}}l l c c c c c}
\hline\hline
& & \multicolumn{2}{c}{Median $J$} & & Below & Eff. rank \\
Sys. & Attack & Normal & Attack & Ratio & $\tau_{\mathrm{lo}}$ & Normal / Attack \\
\hline
\multirow{3}{*}{14-bus} & DNN & 40.2 & 7.7 & 5.2$\times$ & 0.90 & 7.7 / 5.1 \\
 & AE & 40.2 & 10.0 & 4.0$\times$ & 0.88 & 7.7 / 2.3 \\
 & RL-GAN & 40.2 & 1.9 & 21.7$\times$ & 0.87 & 7.7 / 1.6 \\
\hline
\multirow{3}{*}{118-bus} & DNN & 354.7 & 16.0 & 22.2$\times$ & 0.96 & 69.1 / 1.9 \\
 & AE & 354.7 & 12.4 & 28.7$\times$ & 0.94 & 69.1 / 2.8 \\
 & RL-GAN & 354.7 & 8.0 & 44.4$\times$ & 0.98 & 69.1 / 1.2 \\
\hline\hline
\end{tabular*}
\end{table}

\subsubsection{Quantification}
Table~\ref{tab:signature} quantifies the residual-space signature through residual energy and effective rank. Across all six configurations, the median attack $J$ is $4.0$ to $44.4$ times smaller than the normal median, and $0.87$--$0.98$ of the attacked test measurements fall below the lower threshold $\tau_{\mathrm{lo}}$ of the two-sided test, confirming a consistent lower-tail departure.

The attacked residuals also exhibit reduced directional dispersion. Their effective rank, $(\sum_i\lambda_i)^2/\sum_i\lambda_i^2$ computed from the eigenvalues $\lambda_i$ of the unweighted residual covariance, counts the directions that effectively carry residual variation, a threshold-free analogue of the number of principal components, and equals $k$ when the variation is spread equally over $k$ directions. It decreases from $7.7$ for normal measurements to $1.6$--$5.1$ on the 14-bus system, and from $69.1$ to $1.2$--$2.8$ on the 118-bus system. Together, these results establish two complementary manifestations of the signature, reduced residual energy and reduced directional dispersion, both consistent with the stochastic variation deficiency characterized in Section~\ref{sec:defense}.

\subsection{Ablation Studies}
\label{subsec:ablation}

All ablations use the held-out test set of 1200 normal and 1200 attacked measurements per configuration, with the detectors of Section~\ref{subsec:exp-setup}.

\subsubsection{Representation}
\label{subsubsec:abl-repr}
\begin{figure}[!t]
\centering
\includegraphics[width=0.85\linewidth]{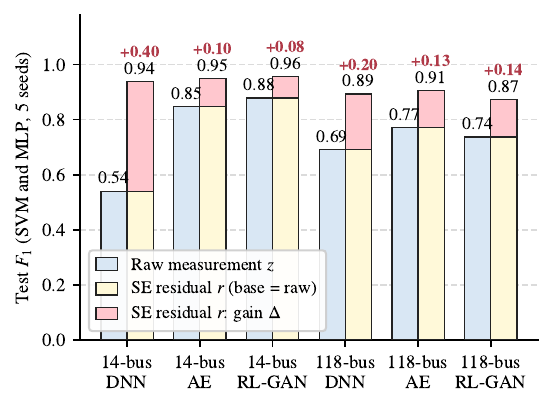}
\caption{Mean test $F_{1}$ of the SVM and MLP on $z$ and $r$ with 100 labels per class.}
\label{fig:s2-vs-s3}
\end{figure}

\begin{table*}[!t]
\caption{Per-classifier test $F_{1}$ on $z$ and $r$ with 100 labels per class.}
\label{tab:s2-vs-s3-detail-100}
\centering
\renewcommand{\arraystretch}{1.10}
\footnotesize
\begin{tabular}{l l | c c c | c c c | c c c}
\hline\hline
\multirow{2}{*}{Sys.} & \multirow{2}{*}{Attack}
 & \multicolumn{3}{c|}{MLP} & \multicolumn{3}{c|}{SVM} & \multicolumn{3}{c}{XGBoost} \\
& & $z$ & $r$ & $\Delta$ & $z$ & $r$ & $\Delta$ & $z$ & $r$ & $\Delta$ \\
\hline
\multirow{3}{*}{14-bus}
 & DNN    & $0.572_{\pm.020}$ & $0.892_{\pm.021}$ & $+.32$ & $0.509_{\pm.055}$ & $0.983_{\pm.004}$ & $+.47$ & $0.747_{\pm.022}$ & $0.932_{\pm.017}$ & $+.18$ \\
 & AE     & $0.958_{\pm.010}$ & $0.906_{\pm.015}$ & $-.05$ & $0.736_{\pm.038}$ & $0.993_{\pm.006}$ & $+.26$ & $0.967_{\pm.010}$ & $0.955_{\pm.025}$ & $-.01$ \\
 & RL-GAN & $0.953_{\pm.011}$ & $0.928_{\pm.016}$ & $-.03$ & $0.803_{\pm.012}$ & $0.987_{\pm.004}$ & $+.18$ & $0.917_{\pm.048}$ & $0.974_{\pm.017}$ & $+.06$ \\
\hline
\multirow{3}{*}{118-bus}
 & DNN    & $0.765_{\pm.014}$ & $0.791_{\pm.125}$ & $+.03$ & $0.616_{\pm.088}$ & $0.995_{\pm.001}$ & $+.38$ & $0.989_{\pm.009}$ & $0.976_{\pm.010}$ & $-.01$ \\
 & AE     & $0.826_{\pm.022}$ & $0.822_{\pm.125}$ & $+.00$ & $0.718_{\pm.068}$ & $0.990_{\pm.002}$ & $+.27$ & $0.987_{\pm.010}$ & $0.967_{\pm.027}$ & $-.02$ \\
 & RL-GAN & $0.836_{\pm.017}$ & $0.755_{\pm.088}$ & $-.08$ & $0.638_{\pm.076}$ & $0.993_{\pm.002}$ & $+.35$ & $0.960_{\pm.023}$ & $0.944_{\pm.054}$ & $-.02$ \\
\hline\hline
\end{tabular}
\end{table*}

We train the same classifiers on raw measurements $z$ and on SE residuals $r$ under the same label budget. Fig.~\ref{fig:s2-vs-s3} reports the mean test $F_1$ of SVM and MLP classifiers using 100 labels per class, averaged over five random seeds. Table~\ref{tab:s2-vs-s3-detail-100} provides classifier-specific results, including XGBoost.

All six attacks are detectable from residuals. Averaged over the SVM and MLP (Fig.~\ref{fig:s2-vs-s3}), the $F_1$ on $r$ exceeds that on $z$ in every configuration, by $0.08$ to $0.40$, rising from $0.54$--$0.88$ on raw measurements to $0.87$--$0.96$ on residuals. Averaged over the six configurations (Table~\ref{tab:s2-vs-s3-detail-100}), $r$ also improves every classifier. The $F_1$ increases from $0.67$ to $0.99$ for SVM, from $0.82$ to $0.85$ for MLP, and from $0.93$ to $0.96$ for XGBoost.

Raw measurements also carry attack information, but its detectability depends on the classifier and the attack. On $z$, XGBoost reaches at least $0.960$ in four of the six configurations, whereas SVM stays between $0.51$ and $0.80$, and no classifier detects all six attacks reliably. On $r$, a standard SVM reaches at least $0.983$ for every attack and system, so the residual-space signature lets a simple classifier achieve uniformly high detection.

\subsubsection{Decision Rule}
\label{subsubsec:abl-rule}

We examine which decision rules can exploit the identified signature. The label-free two-sided test detects $0.89$--$0.99$ of the attacks (Table~\ref{tab:twosided}), demonstrating that the signature can be exploited without labeled attacks or attack-specific training. Supervised residual classifiers also achieve strong performance, with mean $F_1$ scores of $0.87$--$0.96$ (Fig.~\ref{fig:s2-vs-s3}).

In contrast, standard outlier detectors fail to exploit the signature in the full residual vector, even though the supervised classifiers detect the attacks from the same vector. The IF detector applied to $r$ detects at most $0.005$ of the attacks, while one-class SVM and autoencoder detectors detect at most $0.06$, comparable to the $5\%$ of normal snapshots flagged by each detector. At their operating thresholds, these detectors therefore flag attacked measurements no more often than normal ones, apart from a marginal excess for the one-class SVM on the 14-bus RL-GAN attack, as anticipated in Section~\ref{subsec:defense}. Statistical distinguishability alone thus does not ensure detection, and the decision rule must match the direction of the signature.

\begin{table}[!t]
\caption{Detection rates under different measurement noise.}
\label{tab:noise}
\centering
\renewcommand{\arraystretch}{1.10}
\footnotesize
\begin{tabular*}{\linewidth}{@{\extracolsep{\fill}}l c c c}
\hline\hline
& Two-sided & Supervised & BDD \\
Noise & Det & Det & Det \\
\hline
\multicolumn{4}{l}{\textit{Noise magnitude $\sigma$ (Gaussian)}} \\
0.5\% & 0.977 & 0.989 & 0.007 \\
1\% & 0.998 & 0.997 & 0.000 \\
1.5\% & 0.997 & 0.991 & 0.000 \\
2\% & 1.000 & 0.999 & 0.000 \\
3\% & 1.000 & 0.998 & 0.000 \\
\hline
\multicolumn{4}{l}{\textit{Noise shape ($\sigma = 1\%$)}} \\
Gaussian (ref.) & 0.998 & 0.997 & 0.000 \\
Laplace & 0.853 & 0.992 & 0.046 \\
Student-$t_{3}$ & 0.756 & 0.989 & 0.043 \\
Uniform & 0.882 & 0.994 & 0.055 \\
Correlated & 0.704 & 0.988 & 0.050 \\
\hline\hline
\end{tabular*}
\end{table}

\subsubsection{Noise Magnitude and Shape}
\label{subsubsec:abl-noise}

We assess the robustness of the two-sided test and of a supervised detector on the full residual vector $r$ to measurement noise using the 14-bus AE attack, retrained under each noise condition (Table~\ref{tab:noise}). The supervised detection rate is the mean of the SVM and MLP classifiers in Fig.~\ref{fig:s2-vs-s3}.

\textit{Noise magnitude.} Under Gaussian noise levels of $0.5\%$--$3\%$, the test achieves detection rates of $0.977$--$1.000$ at a false-alarm rate of $0.052$. Its performance is therefore stable across the tested noise magnitudes.

\textit{Noise distribution.} We consider Laplace, Student-$t_3$, uniform, and correlated noise at a nominal level of $1\%$, with the Gaussian $1\%$ case of the magnitude sweep as the reference. These conditions are generated as a separate dataset, with the attack retrained for each. The two-sided test achieves detection rates of $0.704$--$0.882$, with false-alarm rates of $0.043$--$0.064$. Heavy-tailed and correlated noise yield the lowest detection rates, but the test still detects at least $0.70$ of the attacks under all four conditions. In contrast, conventional BDD exhibits a false-alarm rate of $0.220$ under correlated noise, reflecting the limitation of its nominal $\chi^2$ threshold under non-Gaussian or correlated noise. Calibrating the two-sided thresholds on normal data from each noise condition keeps the false-alarm rate near the target across the tested distributions.

The supervised detector on $r$, trained with only $100$ labeled samples per class, achieves detection rates of $0.988$--$0.999$ across all tested noise conditions, albeit with higher false-alarm rates of $0.062$--$0.098$. The two-sided test detects at least $0.977$ of the attacks under Gaussian noise of every tested magnitude, but its detection rate falls to $0.704$--$0.882$ under non-Gaussian and correlated noise. The supervised detector therefore complements the two-sided test when even a small number of labeled attacks is available. It recovers the detections lost under these noise shapes, at the cost of a moderately higher false-alarm rate. Both residual-space detectors still far exceed BDD, which detects at most $0.055$ of the attacks.

\section{Conclusion}
\label{sec:conclusion}

This paper identifies a previously overlooked residual-space signature of reconstruction-based DL attacks in AC state estimation. We establish the conditions under which it emerges, a deficiency in residual-space variation manifested as reduced residual energy and directional dispersion. We trace its origin to conditional-mean reconstruction, which suppresses unpredictable meter-noise components that SE exposes by removing operating-state variation. When residual bias is small, this deficiency can reduce BDD detection to or below its normal-data false-alarm rate. We demonstrate that the signature supports both label-free detection using normal data alone and supervised detection using full residual vectors.

Experiments on the IEEE 14-bus and 118-bus systems with DNN, AE, and RL-GAN attacks consistently validate these findings. While BDD detects at most $5\%$ of the attacks, the label-free two-sided test detects up to $99\%$ at about the same false-alarm rate on normal measurements. Ablation studies further examine the effects of residual representation, decision rule, and measurement-noise conditions. Future work will investigate adaptive attackers that explicitly restore normal residual statistics, develop label-free detectors on the full residual vector, and extend the analysis to changing topologies.

\section*{Acknowledgments}
This material is based upon work supported by the US National Science Foundation under Grant No. 2146156. The authors used AI tools (OpenAI ChatGPT, Anthropic Claude, and OpenAI Codex) to assist with writing, language editing, and figure and table preparation. All technical content was verified by the authors.

\bibliographystyle{IEEEtran}
\bibliography{references}

\vfill

\end{document}